\documentclass{lmcs} 
\pdfoutput=1
\usepackage[utf8]{inputenc}

\usepackage{lastpage}
\lmcsdoi{19}{3}{6}
\lmcsheading{}{\pageref{LastPage}}{}{}%
{Jul.~20,~2022}{Jul.~14,~2023}{}


\usepackage{hyperref}
\usepackage{amstext}
\usepackage{amssymb}
\usepackage{comment}
\theoremstyle{plain} 

\newcommand{\ols}[1]{\mskip.5\thinmuskip\overline{\mskip-.5\thinmuskip {#1} \mskip-.5\thinmuskip}\mskip.5\thinmuskip} 
\newcommand{\olsi}[1]{\,\overline{\!{#1}}} 
\makeatletter
\newcommand\closure[1]{
	\tctestifnum{\count@stringtoks{#1}>1} 
	{\ols{#1}} 
	{\olsi{#1}} 
}
\long\def\count@stringtoks#1{\tc@earg\count@toks{\string#1}}
\long\def\count@toks#1{\the\numexpr-1\count@@toks#1.\tc@endcnt}
\long\def\count@@toks#1#2\tc@endcnt{+1\tc@ifempty{#2}{\relax}{\count@@toks#2\tc@endcnt}}
\def\tc@ifempty#1{\tc@testxifx{\expandafter\relax\detokenize{#1}\relax}}
\long\def\tc@earg#1#2{\expandafter#1\expandafter{#2}}
\long\def\tctestifnum#1{\tctestifcon{\ifnum#1\relax}}
\long\def\tctestifcon#1{#1\expandafter\tc@exfirst\else\expandafter\tc@exsecond\fi}
\long\def\tc@testxifx{\tc@earg\tctestifx}
\long\def\tctestifx#1{\tctestifcon{\ifx#1}}
\long\def\tc@exfirst#1#2{#1}
\long\def\tc@exsecond#1#2{#2}

\begin{document}

\title[Second-order revised Krom logic]{Capturing the polynomial hierarchy by\texorpdfstring{\\}{ }second-order revised Krom logic}
\thanks{Corresponding author: Shiguang Feng.}


\author[K.~Wang]{Kexu Wang}[a,*]
\author[S.~Feng]{Shiguang Feng\lmcsorcid{0000-0002-5110-3881}}[b,*]\thanks{$^*$ These authors contributed equally to this work.}
\author[X.~Zhao]{Xishun Zhao}[a]

\address{Institute of Logic and Cognition, Sun Yat-sen University, Guangzhou, 510275, China}	

\address{School of Computer Science and Engineering, Sun Yat-sen University, Guangzhou, 510006, China}	
\email{fengshg3@mail.sysu.edu.cn}  




\begin{abstract}
  \noindent We study the expressive power and complexity of second-order revised Krom logic (SO-KROM$^{r}$). On ordered finite structures, we show that its existential fragment $\Sigma^1_1$-KROM$^r$ equals $\Sigma^1_1$-KROM, and captures NL. On all finite structures, for $k\geq 1$, we show that $\Sigma^1_{k}$ equals $\Sigma^1_{k+1}$-KROM$^r$ if $k$ is even, and $\Pi^1_{k}$ equals $\Pi^1_{k+1}$-KROM$^r$ if $k$ is odd. The results give an alternative logic to capture the polynomial hierarchy. We also introduce an extended version of second-order Krom logic (SO-EKROM). On ordered finite structures, we prove that SO-EKROM collapses to $\Pi^{1}_{2}$-EKROM and equals $\Pi^1_1$. Both SO-EKROM and $\Pi^{1}_{2}$-EKROM capture co-NP on ordered finite structures.
\end{abstract}

\maketitle

\section*{Introduction}

Descriptive complexity studies the logical characterization of computational complexity classes. It describes the property of a problem using the logical method. Computational complexity considers the computational resources such as time and space needed to decide a problem, whereas descriptive complexity explores the minimal logic that captures a complexity class. We say that a logic $\mathcal{L}$ captures a complexity class $\mathcal{C}$, if (i) the data complexity of $\mathcal{L}$ is in $\mathcal{C}$, i.e., for every $\mathcal{L}$ formula $\varphi$, the set of models of $\varphi$ is decidable in $\mathcal{C}$; and (ii) if a class of finite structures is in $\mathcal{C}$, then it is definable by an $\mathcal{L}$ formula. Moreover, if two logics $\mathcal{L}_1$ and $\mathcal{L}_2$ capture two complexity classes $\mathcal{C}_1$ and $\mathcal{C}_2$, respectively, then $\mathcal{L}_1$ and $\mathcal{L}_2$ have the same expressive power if and only if $\mathcal{C}_1$ is equal to $\mathcal{C}_2$~\cite{ebbinghaus1995}.
So the equivalence problem between different complexity classes can be transformed into the expressive power problem of different logics.
In 1974, Fagin showed that the existential fragment of second-order logic ($\exists$SO) captures NP~\cite{fagin1974}. This seminal work had been followed by many studies in the logical characterization of complexity classes. In 1982, Immerman and Vardi independently showed that the least fixed-point logic FO(LFP) captures P on ordered finite structures~\cite{immerman1982relational,vardi1982complexity}. In 1987, Immerman showed that the deterministic transitive closure logic FO(DTC) and transitive closure logic FO(TC) capture L and NL on ordered finite structures, respectively~\cite{immerman1987languages}. In 1989, Abiteboul and Vianu showed that the partial fixed-point logic FO(PFP) captures PSPACE on ordered finite structures~\cite{abiteboul1989fixpoint}. 

Whether P equals NP is an important problem in theoretical computer science. $\exists$SO captures NP on all finite structures. Hence, no logic capturing P on all finite structures implies $\mathrm{P}\neq \mathrm{NP}$. 
The capturing result of FO(LFP) for P is on ordered structures. Actually, FO(LFP) even cannot express the parity of a structure~\cite{ebbinghaus1995}. 
So finding a logic that can capture P effectively on all finite structures is of importance.
Many extensions of FO(LFP) had been studied. FO(IFP,\,\#) is obtained by adding counting quantifiers to the inflationary fixed-point logic FO(IFP) which has the same expressive power as FO(LFP)~\cite{gradel1992inductive,otto1996expressive}. FO(IFP,rank) is an extension of FO(IFP) with the rank operator that can define the rank of a matrix~\cite{dawar2008descriptive,dawar2009logics,atserias2009affine}. Both FO(IFP,\,\#) and FO(IFP,rank) are strictly more expressive than FO(LFP), but neither of them captures P on all finite structures~\cite{ebbinghaus1995,dawar2019approx}. Second-order logic and its fragments are further candidates of logics for P.
In~\cite{gradel1991expressive}, Gr\"{a}del showed that SO-HORN captures P on ordered finite structures. Feng and Zhao introduced second-order revised Horn logic (SO-HORN$^r$) and showed that it equals FO(LFP) on all finite structures~\cite{feng2012,feng2013}.

Similar to the results for P, it is easy to check that no logic capturing NL on all finite structures implies $\mathrm{NL} \neq \mathrm{NP}$.  Gr\"{a}del showed that SO-KROM captures NL on ordered finite structures~\cite{gradel1992}. Cook and Kolokolova introduced the second-order theory V-Krom of bounded arithmetic for NL that is based on SO-KROM~\cite{cook2004}. 
In this paper, we introduce second-order revised Krom logic (SO-KROM$^r$). It is an extension of SO-KROM by allowing the formula $\exists \bar{z} R\bar{z}$ in the clauses where $R$ is a second-order variable. SO-KROM$^r$ is strictly more expressive than SO-KROM. Its existential fragment $\Sigma^1_{1}$-KROM$^r$ is equivalent to SO-KROM on ordered finite structures.
For all $k\geq 1$, on all finite structures, we show that every $\Sigma^1_{k}$ formula is equivalent to a $\Sigma^1_{k+1}$-KROM$^r$ formula for even $k$,  and every $\Pi^1_{k}$ formula is equivalent to a $\Pi^1_{k+1}$-KROM$^r$ formula for odd $k$. Hence, every second-order formula is equivalent to an SO-KROM$^r$ formula. For the data complexity of SO-KROM$^{r}$, we show that $\Sigma^1_{k+1}$-KROM$^r$ is in $\Sigma^p_k$ for even $k$, and $\Pi^1_{k+1}$-KROM$^r$ is in $\Pi^p_k$ for odd $k$,  where $\Sigma^p_0 = \Pi^p_0 = \mathrm{P}$, $\Sigma^p_{k+1}$ is the set of decision problems solvable in nondeterministic polynomial time by a Turing machine augmented with an oracle in $\Sigma^p_{k}$, and $\Pi^p_{k+1}$ is the complement of $\Sigma^p_{k+1}$~\cite{STOCKMEYER1976}. The polynomial time hierarchy $\mathrm{PH} = \bigcup^{\infty}_{k=0} \Sigma^p_{k}$, which is contained within PSPACE. It is well-known that the second-order formulas $\Sigma^1_{k}$ (resp., $\Pi^1_{k}$) capture $\Sigma^p_{k}$ (resp., $\Pi^p_{k}$) ($k\geq 1$)~\cite{Immerman1998descrip}. Combining these we see that SO-KROM$^r$ gives an alternative logical characterization for PH, which is an interesting result in the field of descriptive complexity. The main results in the paper are summarized in Figure~\ref{fig-results}.

\begin{figure}[h]
	\begin{center}
		\begin{tikzpicture}
			\tikzset{expren/.style={thin,rounded corners=8pt,dashed}}
			\tikzset{dacomp/.style={thin,rounded corners=8pt}}
			\draw (0,1) node(sKrom) {$\Sigma^1_1$-KROM};
			\draw (2.2,1) node(Sig2) {$\Sigma^1_2$};
			\draw (4.4,1) node(Sig4) {$\Sigma^1_4$};
			\draw (7.7,1) node(Sig2k) {$\Sigma^1_{2k}$};
			\draw (5.9,0) node(dot1) {$\dots$};
			\draw (9.7,0) node(dot2) {$\dots$};
			\draw[expren] (-1.1,-0.3) rectangle (1,1.4);
			\draw[expren] (1.2,-0.3) rectangle (3.2,1.4);
			\draw[expren] (3.4,-0.3) rectangle (5.5,1.4);
			\draw[expren] (6.25,-0.3) rectangle (9.15,1.4);
			
			\draw (0,0) node(Sigr12) {$\Sigma^{rk}_1\equiv \Sigma^{rk}_2$};
			\draw (2.3,0) node(Sigr34) {$\Sigma^{rk}_3\equiv \Sigma^{rk}_4$};
			\draw (4.5,0) node(Sigr56) {$\Sigma^{rk}_5\equiv \Sigma^{rk}_6$};
			\draw (7.7,0) node(Sigr2k) {$\Sigma^{rk}_{2k+1}\equiv \Sigma^{rk}_{2k+2}$};
			\draw[dacomp] (-1.1,-1.3) rectangle (1,0.4);
			\draw[dacomp] (1.2,-1.3) rectangle (3.2,0.4);
			\draw[dacomp] (3.4,-1.3) rectangle (5.5,0.4);
			\draw[dacomp] (6.25,-1.3) rectangle (9.15,0.4);
			
			\draw (0,-1) node(NL) {NL};
			\draw (2.3,-1) node(Sigp2) {$\Sigma^p_2$};
			\draw (4.5,-1) node(Sigp4) {$\Sigma^p_4$};
			\draw (7.8,-1) node(Sigp2k) {$\Sigma^p_{2k}$};
			
			\draw (0,-2.2) node(pKrom) {$\Pi^1_1$-KROM};
			\draw (2.25,-2.2) node(Pi1) {$\Pi^1_1$};
			\draw (4.45,-2.2) node(Pi3) {$\Pi^1_3$};
			\draw (7.7,-2.2) node(Sig8) {$\Pi^1_{2k-1}$};
			\draw (5.9,-3.2) node(dot3) {$\dots$};
			\draw (9.7,-3.2) node(dot4) {$\dots$};
			\draw[expren] (-1.1,-3.5) rectangle (1,-1.8);
			\draw[expren] (1.2,-3.5) rectangle (3.2,-1.8);
			\draw[expren] (3.4,-3.5) rectangle (5.5,-1.8);
			\draw[expren] (6.25,-3.5) rectangle (9.15,-1.8);
			
			\draw (0,-3.2) node(Pir12) {$\Pi^{rk}_1$};
			\draw (2.3,-3.2) node(Pi23) {$\Pi^{rk}_2\equiv \Pi^{rk}_3$};
			\draw (4.5,-3.2) node(Pi45) {$\Pi^{rk}_4\equiv \Pi^{rk}_5$};
			\draw (7.8,-3.2) node(Pi2k) {$\Pi^{rk}_{2k}\equiv \Pi^{rk}_{2k+1}$};
			\draw[dacomp] (-1.1,-4.5) rectangle (1,-2.8);
			\draw[dacomp] (1.2,-4.5) rectangle (3.2,-2.8);
			\draw[dacomp] (3.4,-4.5) rectangle (5.5,-2.8);
			\draw[dacomp] (6.25,-4.5) rectangle (9.15,-2.8);
			
			\draw (0,-4.2) node(emy) {};
			\draw (2.25,-4.2) node(Pi1) {$\Pi^p_1$};
			\draw (4.45,-4.2) node(Pi3) {$\Pi^p_3$};
			\draw (7.7,-4.2) node(Pip2k) {$\Pi^p_{2k-1}$};
		\end{tikzpicture}
	\end{center}
	\caption{\label{fig-results} The expressive power and complexity of SO-KROM$^r$. $\Sigma^{rk}_k$ and $\Pi^{rk}_k$ denote $\Sigma^1_k$-KROM$^{r}$ and $\Pi^1_k$-KROM$^{r}$, respectively. The dashed rectangle parts show the equivalence relation between second-order formulas and SO-KROM$^{r}$ formulas. The solid rectangle parts show the capturing results of SO-KROM$^{r}$ for PH.}
\end{figure}
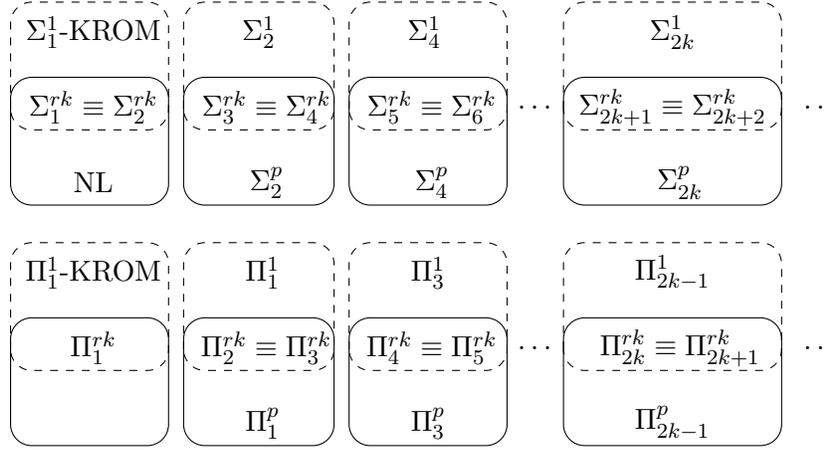

The paper is organized as follows. In Section~\ref{sec-pre}, we give the basic definitions and notations. In Section~\ref{sec-descrp-exst-kr}, we study the expressive power and complexity of the existential fragment of SO-KROM$^{r}$. In Section~\ref{sec-expre-sokr}, we study the descriptive complexity of SO-KROM$^{r}$. In Section~\ref{sec-exspkr}, we introduce second-order extended Krom logic and study its descriptive complexity. Section~\ref{sec-concln} is the conclusion of the paper.

\section{Preliminaries}\label{sec-pre}
Let $\tau=\{\mathbf{c}_{1},\mathbf{c}_{2},\dots,\mathbf{c}_{m},P_{1},P_{2},\dots,P_{n}\}$ be a vocabulary, where $\mathbf{c}_{1},$ $\mathbf{c}_{2},\dots,\mathbf{c}_{m}$ are constant symbols and $P_{1},P_{2},\dots,P_{n}$ are relation symbols. 
A \textit{$\tau$-structure} $\mathcal{A}$ is a tuple 
\[\langle A,\mathbf{c}_{1}^{A},\mathbf{c}_{2}^{A},\dots,\mathbf{c}_{m}^{A},P_{1}^{A},P_{2}^{A},\dots,P_{n}^{A}\rangle,\]
where $A$ is the domain of $\mathcal{A}$, and $\mathbf{c}_{1}^{A},\mathbf{c}_{2}^{A},\dots,\mathbf{c}_{m}^{A}$, $P_{1}^{A},P_{2}^{A},\dots,P_{n}^{A}$ are the interpretations of the constant and relation symbols over $A$, respectively. We assume the identity relation ``='' is contained in every vocabulary, and omit the superscript ``$A$'' in the notation when no confusion is caused.
We call $\mathcal{A}$ finite if its domain $A$ is a (nonempty) finite set. In this paper, all structures considered are finite. We use $|\ |$ to denote the cardinality of a set or the arity of a tuple, e.g., $|\{a,b,c\}|=3$ and $|\bar{x}|=3$ where $\bar{x}=(x_1,x_2,x_3)$, and $arity(X)$ to denote the arity of a relation symbol (variable) $X$.
A finite structure is \textit{ordered} if it is equipped with a linear order relation ``$\leq$'', a successor relation ``$\mathrm{SUCC}$'', and constants ``\textbf{min}'' and ``\textbf{max}'' interpreted as the minimal and maximal elements, respectively.

Given a logic $\mathcal{L}$, we use $\mathcal{L}(\tau)$ to denote the set of $\mathcal{L}$ formulas over vocabulary $\tau$. For better readability, the symbol ``$\tau$'' is omitted when it is clear from context. Given two logics $\mathcal{L}_1$ and $\mathcal{L}_2$, we use $\mathcal{L}_1\leq \mathcal{L}_2$ to denote that every $\mathcal{L}_1$ formula is equivalent to an $\mathcal{L}_2$ formula. If both $\mathcal{L}_1\leq \mathcal{L}_2$ and $\mathcal{L}_2\leq \mathcal{L}_1$ hold, then we write $\mathcal{L}_1\equiv \mathcal{L}_2$. 

\begin{defi}
	Given a vocabulary $\tau$, the second-order Krom logic over $\tau$, denoted by SO-KROM($\tau$), is a set of second-order formulas of the form
	\[
	Q_{1}R_{1}\cdots Q_{m}R_{m}\forall\bar{x}(C_{1}\wedge\cdots\wedge C_{n})
	\]
	where each $Q_{i}\in\{\forall,\exists\}$, $C_1,\dots,C_n$ are Krom clauses with respect to $R_{1},\dots,R_{m}$,
	more precisely, each $C_{j}$ is a disjunction of the form
	\[
	\beta_{1}\vee\cdots\vee\beta_{q}\vee H_{1}\vee H_{2},
	\]
	where
	\begin{enumerate}
		\item each $\beta_{s}$ for $s \in \{1,\dots, q\}$ is either $P\bar{y}$ or $\neg P\bar{y}$ $(P \in \tau)$;
		\item each $H_{t}$ is either $R_{i}\, \bar{z}$, $\neg R_{i}\, \bar{z}$ $(1\leq i\leq m)$, or $\bot$ (for false).
	\end{enumerate}
	If we replace (2) by
	\begin{enumerate}
		\item[(2$'$)] each $H_{t}$ is either $R_{i}\, \bar{z}$, $\neg R_{i}\, \bar{z}$, $\exists z_1\cdots \exists z_{arity(R_i)}\, R_{i}\, z_1\dots z_{arity(R_i)}$ $(1\leq i\leq m)$, or $\bot$ (for false),
	\end{enumerate}
	then we call this logic second-order revised Krom Logic, denoted by SO-KROM$^{r}(\tau)$. 
\end{defi}
We use $\Sigma_{k}^{1}\text{-KROM}^{r}$ (resp., $\Pi_{k}^{1}\text{-KROM}^{r}$) to denote the set of SO-KROM$^{r}$ formulas whose second-order prefix starts with an existential (resp., a universal) quantifier and alternates $k-1$ times between series of existential and universal quantifiers.

\begin{exa}
	A directed graph is strongly connected iff there exists a path between every pair of nodes. The strong connectivity problem is NL-complete, which is defined by
	\begin{description}
		\item [{Input}] a directed graph $G=(V,E)$,
		\item [{Output}] yes if $G$ is strongly connected, and no otherwise.
	\end{description}
	Since $\mathrm{NL} = \mathrm{co\text{-}NL}$, the complement of the strong connectivity problem is also NL-complete, which can be defined by the following $\Sigma_{1}^{1}\text{-KROM}^{r}$ formula
	\[\exists R\exists Y \forall x\forall y\forall z \left( 
	\begin{array}{l}
		(Exy\rightarrow Rxy)\wedge(Exy\wedge Ryz\rightarrow Rxz) \\
		\wedge (\neg Rxy \leftrightarrow Yxy) \wedge \exists u\exists v Yuv
	\end{array}\right) 
	\]
	where $R$ is the transitive closure of $E$, and $Y$ is the complement of $R$. A graph $G$ satisfies the formula iff there exist two nodes $a,b$ such that $a$ cannot reach $b$.
\end{exa}

SO-KROM is closed under substructures~\cite{gradel1992}, which means that if a structure satisfies a SO-KROM formula then all its substructures also satisfy the formula.
Because a non-strongly connected graph may be made strongly connected by removing nodes, SO-KROM cannot define the complement of the strong connectivity problem. The above example shows that SO-KROM$^{r}$ is strictly more expressive than SO-KROM.

\section{The expressive power and complexity of $\Sigma_{1}^{1}$-KROM$^{r}$}\label{sec-descrp-exst-kr}

In this section, we study the expressive power of the universal and existential fragments of SO-KROM$^{r}$,  and show that $\Sigma_{1}^{1}$-KROM$^{r}$ captures NL on finite ordered structures.

\begin{prop}\label{prop-Pikroeqfirt}
	Every $\Pi^1_1\text{-}\mathrm{KROM}^{r}$ formula is equivalent to a first-order formula $\forall \bar{x} \varphi$, where $\varphi$ is a quantifier-free $\mathrm{CNF}$ formula.
\end{prop}
\begin{proof}
	Given a $\Pi^1_1$-KROM$^{r}$ formula $\Phi= \forall X_1\dots \forall X_n \forall \bar{x}(C_1\wedge \dots \wedge C_m)$, we deal with each clause $C_j$ for $j \in \{1,\dots, m\}$ as follows. Let $\alpha$ denote the first-order part of $C_j$.
	\begin{description}
		\item[Case\;1] If $C_j= \alpha \vee \neg X_i \bar{x}_1 \vee \exists \bar{x}_2 X_i \bar{x}_2$, then remove the clause $C_j$.
		\item[Case\;2] If  $C_j= \alpha \vee X_i \bar{x}_1 \vee \neg X_i \bar{x}_2$, then replace $C_j$ by $\alpha \vee \bar{x}_1 = \bar{x}_2$.
		\item[Case\;3] For the other cases, remove all occurrences of second-order variables in $C_j$.
	\end{description}
	After the above steps, all second-order variables are removed and we obtain a first-order formula $\phi=\forall \bar{x}(C_1'\wedge \dots \wedge C_{m'}')$, where each $C_j'$ is quantifier-free. 
	
	In \textbf{Case\;1}, $\neg X_i \bar{x}_1 \vee \exists \bar{x}_2 X_i \bar{x}_2$ is a tautology, so the clause can be removed safely. 
	In \textbf{Case\;2}, if $\bar{x}_1 = \bar{x}_2$, then $X_i \bar{x}_1 \vee \neg X_i \bar{x}_2$ is always true; if $\bar{x}_1 \neq \bar{x}_2$, then there exists a valuation for $X_i$ such that $X_i \bar{x}_1 \vee \neg X_i \bar{x}_2$ is false, the clause is true iff $\alpha$ is true. Hence, $\forall X_i \forall\bar{x}(\alpha \vee X_i \bar{x}_1 \vee \neg X_i \bar{x}_2)$ is equivalent to $\forall\bar{x}(\alpha \vee \bar{x}_1 = \bar{x}_2)$. In \textbf{Case\;3}, there is always a valuation for the second-order variables in $C_j$ under which the second-order part of $C_j$ is false. So all occurrences of second-order variables can be removed from $C_j$.
	Therefore, $\Phi$ and $\phi$ are equivalent.
\end{proof}

\begin{cor}\label{coro-sokequiv}
	For each $k\geq 1$, if $k$ is odd, then $\Sigma^1_k\text{-}\mathrm{KROM}^{r}\equiv \Sigma^1_{k+1}\textbf{-}\mathrm{KROM}^{r}$; and if $k$ is even, then $\Pi^1_k\text{-}\mathrm{KROM}^{r}\equiv \Pi^1_{k+1}\text{-}\mathrm{KROM}^{r}$.
\end{cor}
\begin{proof}
	If the type of the innermost second-order quantifier block of a SO-KROM$^{r}$ formula is universal, then it can be removed by Proposition~\ref{prop-Pikroeqfirt} to get an equivalent formula.
\end{proof}

We use $\varphi[\alpha/\beta]$ to denote replacing the variable (or the formula) $\alpha$ in $\varphi$ with $\beta$.

\begin{lem}\label{lem-exists}
	Let $\exists x_1\dots \exists x_n \phi$ be a quantified Boolean formula. It is equivalent to the following formula
	\[\phi[x_1/\bot,\dots, x_n/\bot]\vee \bigvee_{1\leq i\leq n}\exists x_1\dots \exists x_{i-1}\exists x_{i+1}\dots\exists x_n\phi[x_i/\top].\]
\end{lem}
\begin{proof}
	$\exists x_1\dots \exists x_n \phi$ is true iff $\phi$ is true when all $x_1,\dots,x_n$ are false, or for some $x_i$, where $(1\leq i \leq n)$, the formula 
	$\exists x_1\dots \exists x_{i-1}\exists x_{i+1}\dots\exists x_n\phi$ is true when $x_i$ is true.
\end{proof}

From Lemma~\ref{lem-exists} we can infer the following proposition. 
\begin{prop}\label{prop-sig2krom}
	Every $\Sigma^1_1\text{-}\mathrm{KROM}^r$ formula is equivalent to a formula of the form $\exists \bar{y}_1 \phi_1\vee \dots \vee \exists \bar{y}_n \phi_n$,
	where each $\phi_i$ for $i\in \{1,\dots,n\}$ is a $\Sigma^1_1\text{-}\mathrm{KROM}$ formula.
\end{prop}
\begin{proof}
	Let	$\Psi=\exists R \exists \ols{Y} \forall \bar{x} \phi$ be a $\Sigma^1_1$-KROM$^r$ formula, and $\alpha(\bar{z})=(\bar{z}=\bar{y})\vee R\bar{z}$ where $\bar{y}$ have no occurrence in $\phi$.
	It is easily seen that if $\bar{z}=\bar{y}$ holds then $\alpha(\bar{z})$ is true, and if $\bar{z}\neq \bar{y}$ holds then $\alpha(\bar{z})$ is equivalent to $R\bar{z}$. So $\alpha(\bar{z})$ is equivalent to $R\bar{z}$ except at the point $\bar{y}$. 
	Define
	\[
	\Psi'= \exists \ols{Y} \forall \bar{x}\phi[R\bar{z}/\bot]\vee \exists\bar{y} \exists R \exists \ols{Y} \forall \bar{x} \phi[R\bar{z}/\alpha(\bar{z})].
	\]
	We show that $\Psi$ and $\Psi'$ are equivalent. It is easily seen that for any structure $\mathcal{A}$, $\mathcal{A}\models \Psi$ iff $(\mathcal{A},R)\models \exists \ols{Y} \forall \bar{x} \phi$, where either  $R=\emptyset$ or $R$ is not empty.
	Every occurrence of $\exists \bar{z}R\bar{z}$ in $\Psi'$ is either replaced by $\exists \bar{z} \bot$ or replaced by $\exists\bar{z}\alpha(\bar{z})$ which is a tautology. We remove the occurrences of $\exists \bar{z} \bot$ and the clauses containing $\exists\bar{z}\alpha(\bar{z})$ in $\Psi'$.
	For any structure $\mathcal{A}$, we can construct a quantified Boolean formula $\Psi_{\mathcal{A}}$ such that $\mathcal{A}\models \Psi$ iff $\Psi_{\mathcal{A}}$ is true (see the proof of Proposition~\ref{prop-exprness} for details of the construction). Similarly, we can construct $\Psi'_{\mathcal{A}}$ such that $\mathcal{A}\models \Psi'$ iff $\Psi'_{\mathcal{A}}$ is true. By Lemma~\ref{lem-exists}, $\Psi_{\mathcal{A}}$ and $\Psi'_{\mathcal{A}}$ are equivalent. Therefore, $\Psi$ and $\Psi'$ are equivalent. The same procedure can be repeated for each $Y_i \in \ols{Y}$ until all occurrences of $\exists \bar{v}Y_i\bar{v}$ are removed. Finally, we can obtain an equivalent formula of the form $\exists \bar{y}_1 \phi_1\vee \dots \vee \exists \bar{y}_n \phi_n$, where each $\phi_i$ $(1\leq i \leq n)$ is a $\Sigma^1_1$-KROM formula.
\end{proof}

\begin{prop}\label{prop-sg1kromNL}
	The data complexity of $\Sigma^1_1\text{-}\mathrm{KROM}^r$ is in $\mathrm{NL}$.
\end{prop}
\begin{proof}
	By Proposition~\ref{prop-sig2krom}, we only need to show that the data complexity of the formula $\exists \bar{y}_1 \phi_1\vee \dots \vee \exists \bar{y}_n \phi_n$, where each $\phi_i$ $(1\leq i \leq n)$ is a $\Sigma^1_1$-KROM formula, is in NL.  Given a structure $\mathcal{A}$, the Turing machine can nondeterministically choose an $i\in \{1,\dots, n\}$ and a tuple $\bar{u}_i\in A^{|\bar{y}_i|}$ in logarithmic space. Whether $\mathcal{A}\models \phi_i[\bar{u}_i]$ holds can be checked in NL since the data complexity of $\Sigma^1_1$-KROM is in NL~\cite{gradel1992}.
\end{proof}

Every $\Sigma^1_1$-KROM formula is also a  $\Sigma^1_1$-KROM$^r$ formula. Because $\Sigma^1_1$-KROM captures NL on ordered finite structures~\cite{gradel1992}, combining Corollary~\ref{coro-sokequiv} we obtain the following corollary.

\begin{cor}\label{coro-sigcapnl}
	Both $\Sigma^1_1\text{-}\mathrm{KROM}^r$ and $\Sigma^1_2\text{-}\mathrm{KROM}^r$ capture $\mathrm{NL}$ on ordered finite structures.
\end{cor}

\section{The descriptive complexity of SO-KROM$^{r}$}\label{sec-expre-sokr}
SO-KROM collapses to its existential fragment. This is unlikely to be true for SO-KROM$^{r}$ by the following result. Let $\Sigma_k$-CNF (resp., $\Sigma_k$-DNF) denote the set of quantified Boolean formulas $\exists\bar{x}_1\forall\bar{x}_2\exists\bar{x}_3\dots Q_k\bar{x}_k\phi$ whose prefix starts with an existential quantifier and has $k-1$ alternations between series of existential and universal quantifiers, and the matrix $\phi$ is a quantifier-free formula in conjunctive normal form (resp., disjunctive normal form). The definitions for $\Pi_k$-CNF and $\Pi_k$-DNF are similar where the formula's prefix starts with a universal quantifier.
Given a set $\mathbf{F}$ of quantified Boolean formulas, the evaluation problem of $\mathbf{F}$ is deciding the truth value of the formulas in it. For the polynomial hierarchy, it is shown that the evaluation problem of $\Sigma_k$-CNF (resp., $\Sigma_k$-DNF) is $\Sigma^p_k$-complete if $k$ is odd (resp., even)~\cite{STOCKMEYER1976}. Hence, the evaluation problem of $\Pi_k$-DNF (resp., $\Pi_k$-CNF) is $\Pi^p_k$-complete if $k$ is odd (resp., even) by duality.

\begin{prop}\label{prop-phINkrom}
	The evaluation problem of $\Sigma_k\text{-}\mathrm{DNF}$ is definable in $\Sigma_{k+1}^{1}\text{-}\mathrm{KROM}^{r}$ if $k$ is even, and the evaluation problem of $\Pi_k\text{-}\mathrm{DNF}$ is definable in $\Pi_{k+1}^{1}\text{-}\mathrm{KROM}^{r}$ if $k$ is odd.
\end{prop}
\begin{proof}
	We only prove for $\Sigma_k$-DNF where $k$ is even, the proof for $\Pi_k$-DNF where $k$ is odd is the same as it. Let vocabulary $\tau=\{\mathrm{Clause},\mathrm{Var}_{1},\dots,\mathrm{Var}_{k},\mathrm{Pos},\mathrm{Neg}\}$, where $\mathrm{Clause}$, $\mathrm{Var}_1,\dots, \mathrm{Var}_k$ are unary relation symbols, and $\mathrm{Pos},\mathrm{Neg}$ are binary relation symbols. Using a similar method as in~\cite{Immerman1998descrip}, we can encode a $\Sigma_k$-DNF formula $\exists\bar{x}_1\forall\bar{x}_2\cdots\exists\bar{x}_{k-1}\forall\bar{x}_{k}\phi$ via a $\tau$-structure $\mathcal{A}$ such that for any $i,j\in A$, $\mathrm{Clause}\,i$ holds iff $i$ is a clause, $\mathrm{Var}_h\,j$ holds iff $j$ is a variable occurring in the quantifier block $\exists(\forall)\bar{x}_h$ for $h\in \{1,\dots, k\}$, and $\mathrm{Pos}\,ij$ (resp., $\mathrm{Neg}\,ij$) holds iff variable $j$ occurs positively (resp., negatively) in clause $i$.
	For example, the $\Sigma_4$-DNF formula
	\[ 
	\exists x_1\forall x_2\exists x_3\forall x_4(\underbrace{(x_1\wedge \neg x_2)}_1\vee \underbrace{(x_2\wedge \neg x_4)}_2\vee \underbrace{(x_3\wedge x_4)}_3) 
	\]
	can be encoded via the structure $\langle \{1,2,3,4\},\mathrm{Clause},\mathrm{Var}_1,\mathrm{Var}_2,\mathrm{Var}_3,\mathrm{Var}_4,\mathrm{Pos},\mathrm{Neg}\rangle$, where
	$\mathrm{Clause}=\{1,2,3\}$, $\mathrm{Var}_1=\{1\}$, $\mathrm{Var}_2=\{2\}$, $\mathrm{Var}_3=\{3\}$, $\mathrm{Var}_4=\{4\}$, $\mathrm{Neg}=\{(1,2),(2,4)\}$, $\mathrm{Pos}=\{(1,1),(2,2),(3,3),(3,4)\}$.
	Let $\Phi$ be the following formula 
	\[
	\exists X_1\forall X_2\dots\exists X_{k-1}\forall X_k\exists Y \forall x\forall y \left(
	\begin{array}{l}
		\exists z Yz\wedge (Yx\rightarrow \mathrm{Clause}\,x)\wedge\\
		\bigwedge_{1\leq h\leq k}(Yx\wedge \mathrm{Pos}\, xy\wedge \mathrm{Var}_{h}\,y\rightarrow X_{h}\,y)\wedge \\
		\bigwedge_{1\leq h\leq k}(Yx\wedge \mathrm{Neg}\,xy\wedge \mathrm{Var}_{h}\,y\rightarrow \neg X_{h}\,y)
	\end{array}\right).
	\]
	Obviously, $\Phi$ is a $\Sigma_{k+1}^{1}$-KROM$^{r}$ formula, and it expresses that there is a valuation $X_1$ to $\bar{x}_1$, for any valuation $X_2$ to $\bar{x}_2,\dots$, there is a valuation $X_{k-1}$ to $\bar{x}_{k-1}$, for any valuation $X_k$ to $\bar{x}_{k}$, there is a nonempty set $Y$ of clauses such that every literal in the clauses in $Y$ is true under the valuation.
	For an arbitrary $\Sigma_k$-DNF formula $\psi$, let $\mathcal{A}$ be the $\tau$-structure that encodes $\psi$, it is easily seen that $\mathcal{A}\models\Phi$ iff $\psi$ is true.
\end{proof}

Before showing that every second-order formula is equivalent to a SO-KROM$^{r}$ formula, we first prove a lemma.

\begin{lem}\label{lem_fistequiv}
	Every first-order formula is equivalent to a second-order formula
	$\exists Y\forall\bar{x}(\exists\bar{y}Y\bar{z}\bar{y}\wedge C_{1}\wedge\cdots\wedge C_{m})$
	where each $C_{i}$ is a disjunction of atomic or negated atomic formulas.
\end{lem}
\begin{proof}
	Given a first-order formula $\varphi$, without loss of generality, assume that $\varphi$ is in the prenex normal form $\forall\bar{x}_{1}\exists\bar{y}_{1}\cdots\forall\bar{x}_{n}\exists\bar{y}_{n}(C_{1}\wedge\cdots\wedge C_{m})$,
	where each $C_{i}$ for $i\in \{1,\dots,m\}$ is a disjunction of atomic or negated atomic formulas.
	Define
	\[
	\begin{array}{cl}
		\varphi_{1}= & \forall\bar{x}_{1}\cdots\forall\bar{x}_{n}\exists\bar{y}_{1}\cdots\exists\bar{y}_{n}Y\bar{x}_{1}\cdots\bar{x}_{n}\bar{y}_{1}\cdots\bar{y}_{n},\\
		\\
		\varphi_{2}= & \forall\bar{x}_{1}\cdots\forall\bar{x}_{n}\forall\bar{y}_{1}\cdots\forall\bar{y}_{n}  \forall\bar{x}_{1}'\cdots\forall\bar{x}_{n}'\forall\bar{y}_{1}'\cdots\forall\bar{y}_{n}'\\
		& \biggl(Y\bar{x}_{1}\cdots\bar{x}_{n}\bar{y}_{1}\cdots\bar{y}_{n}\wedge Y\bar{x}_{1}'\cdots\bar{x}_{n}'\bar{y}_{1}'\cdots\bar{y}_{n}'\\
		& \rightarrow\bigwedge_{1\leq i\leq n}\Bigl((\bigwedge_{1\leq j\leq i}\bar{x}_{j}=\bar{x}_{j}')\rightarrow\bar{y}_{i}=\bar{y}_{i}'\Bigr)\biggr),\\
		\\
		\varphi_{3}= & \forall\bar{x}_{1}\forall\bar{y}_{1}\cdots\forall\bar{x}_{n}\forall\bar{y}_{n}(Y\bar{x}_{1}\cdots\bar{x}_{n}\bar{y}_{1}\cdots\bar{y}_{n}\rightarrow \bigwedge_{1\leq i \leq m} C_i).
	\end{array}
	\]
	The relation $Y$ encodes a Skolem function for each $\bar{y}_i$ $(1\leq i \leq n)$, whose value only depends on the values of $\bar{x}_1,\dots,\bar{x}_i$. It is easy to check that $\varphi$ is equivalent to the formula $\exists Y(\varphi_{1}\wedge\varphi_{2}\wedge\varphi_{3})$, which can be converted to the form of $\exists Y\forall\bar{x}(\exists\bar{y}Y\bar{z}\bar{y}\wedge C_{1}\wedge\cdots\wedge C_{m})$, where $\bar{z}$ are variables from $\bar{x}$.
\end{proof}

\begin{prop}\label{prop-exprness}
	Every second-order formula is equivalent to an $\mathrm{SO}\text{-}\mathrm{KROM}^{r}$ formula. More precisely, for each $k\geq 1$,
	if $k$ is even, then $\Sigma^1_{k} \leq \Sigma^1_{k+1}\text{-}\mathrm{KROM}^r$; and if $k$ is odd, then $\Pi^1_{k} \leq \Pi^1_{k+1}\text{-}\mathrm{KROM}^r$.
\end{prop}
\begin{proof}
	Given a $\Sigma^1_k$-formula $\exists X_1 \forall X_2 \dots \exists X_{k-1} \forall X_k \varphi$, where $k$ is even and $\varphi$ does not contain second-order quantifiers, we show that it is equivalent to a $\Sigma^1_{k+1}$-KROM$^r$ formula. The proof for the other cases is essentially the same. 
	By Lemma~\ref{lem_fistequiv}, $\neg \varphi$ is equivalent to a formula $\exists X_{k+1}\forall\bar{x}(\exists\bar{y}X_{k+1}\bar{z}\bar{y}\wedge C_{1}\wedge\cdots\wedge C_{m})$. So $\exists X_1 \forall X_2 \dots \exists X_{k-1} \forall X_k \varphi$ is equivalent to the formula
	\[
	\Phi = \exists X_1 \forall X_2 \dots \exists X_{k-1} \forall X_k \forall X_{k+1}\exists \bar{x}(\forall\bar{y}\neg X_{k+1}\bar{z}\bar{y}\vee D_{1}\vee \cdots\vee D_{m}) 
	\]
	where each $D_{j}$ is a conjunction of atomic (or negated atomic) formulas.
	Suppose that $\Phi$ is over vocabulary $\sigma$. Given a $\sigma$-structure $\mathcal{A}$, we construct a $\Sigma_{k}$-DNF quantified Boolean formula $\psi$ such that $\mathcal{A}\models \Phi$ iff $\psi$ is true.
	Let $A$ be the domain of $\mathcal{A}$.
	We replace the first-order part $\exists \bar{x}(\forall\bar{y}\neg Y\bar{z}\bar{y}\vee D_{1}\vee \cdots\vee D_{m}) $ by 
	\[ 
	\bigvee_{\bar{a}\in A^{|\bar{x}|}} \biggl( \Bigl(\bigwedge_{\bar{b}\in A^{|\bar{y}|}}\neg Y\bar{z}\bar{y}[\bar{y}/\bar{b}] \Bigr)\vee D_{1}\vee \cdots\vee D_{m} \biggr) [\bar{x}/\bar{a}].
	\]
	We remove the clauses with a formula $(\neg)R\bar{c}$ that is false in $\mathcal{A}$ and delete the formulas $(\neg)R\bar{c}$ that are true in $\mathcal{A}$ in every clause, where $R\in \sigma$. 
	Then we replace each quantifier $\exists X_i$ (or $\forall X_i$) with a sequence $\exists X_i \bar{d}_1\dots \exists X_i \bar{d}_{|A|^{arity(X_i)}}$ (or $\forall X_i \bar{d}_1\dots \forall X_i \bar{d}_{|A|^{arity(X_i)}}$) where each $\bar{d}_j\in A^{arity(X_i)}$. We treat the atoms $X_i \bar{d}_j$ as propositional variables, and the resulting formula $\psi$ is a $\Sigma^1_{k}$-DNF quantified Boolean formula. It is clear that $\mathcal{A}\models \Phi$ iff $\psi$ is true.
	
	By Proposition~\ref{prop-phINkrom} and its proof, we know that $\psi$ can be encoded in a $\tau$-structure $\mathcal{B}$, where $\tau = \langle \mathrm{Clause},\mathrm{Var}_{1},\dots,\mathrm{Var}_{k},\mathrm{Pos},\mathrm{Neg}\rangle$ and there is a $\Sigma_{k+1}^{1}$-KROM$^{r}(\tau)$ formula $\Psi$ such that $\psi$ is true iff $\mathcal{B}\models\Psi$. In the following, we define a quantifier-free interpretation
	\[
	\Pi=\Bigl(\pi_{\mathrm{uni}}(\bar{v}),\pi_{\mathrm{Clause}}(\bar{v}),\pi_{\mathrm{Var}_1}(\bar{v}),\dots,\pi_{\mathrm{Var}_{k}}(\bar{v}),
	\pi_{\mathrm{Pos}}(\bar{v}_{1},\bar{v}_{2}),\pi_{\mathrm{Neg}}(\bar{v}_{1},\bar{v}_{2}) \Bigr)
	\]
	of $\tau$ in $\sigma$, where $\pi_{\mathrm{uni}},\pi_{\mathrm{Clause}},\pi_{\mathrm{Var}_1},\dots,\pi_{\mathrm{Var}_{k}},
	\pi_{\mathrm{Pos}},\pi_{\mathrm{Neg}}$ are all quantifier-free formulas over $\sigma$.
	Intuitively, $\pi_{\mathrm{uni}}$ defines the domain of $\mathcal{B}$, $\pi_{\mathrm{Clause}}$ defines the set of clauses of $\psi$, each $\pi_{\mathrm{Var}_i}$ $(1\leq i \leq k)$ defines the set of variables occurring in the quantifier block $\exists(\forall)X_i$, $\pi_{\mathrm{Pos}}$ (or $\pi_{\mathrm{Neg}}$) defines a variable occurs positively (or negatively) in a clause.
	
	For any $\sigma$-structure $\mathcal{A}$, $\Pi$ defines a $\tau$-structure $\mathcal{A}^{\Pi}$ that encodes the formula $\psi$ such that $\mathcal{A}^{\Pi}\models \Psi$ iff $\psi$ is true iff $\mathcal{A}\models\Phi$.
	Since $\Pi$ is an interpretation of $\tau$ in $\sigma$, we can construct a $\Sigma_{k+1}^{1}$-KROM$^{r}(\sigma)$ formula $\Psi^{-\Pi}$ from $\Psi$ such that $\mathcal{A}^{\Pi}\models \Psi$ iff $\mathcal{A}\models\Psi^{-\Pi}$. Therefore, $\Psi^{-\Pi}$ and $\Phi$ are equivalent. For more details of the interpretation from one vocabulary to another, we refer the reader to~\cite{ebbinghaus1995}.
	
	We suppose that $\mathcal{A}$ contains at least two different elements. Let
	\begin{align*}
		& g  =  \max\{arity(X_{1}),\dots,arity(X_{k}),arity(X_{k+1}))\}, \\
		& d  =  3+\max\{(|\bar{x}|+m+1),(g+k+1)\}.
	\end{align*}
	Define the width of $\Pi$ to be $d$. Let $\pi_{\mathrm{uni}}(\bar{v})=\bigwedge^d_{i=1}(v_{i}=v_{i})$, it defines the domain of $\mathcal{A}^{\Pi}$. 
	For any $\bar{a}=(a_{1},a_{2},\dots,a_{d})\in A^d$, we will make the following assumptions:
	\begin{itemize}
		\item if $\bar{a}$ encodes a clause, then $a_1\neq a_3\wedge a_2 = a_3$, and
		\item if $\bar{a}$ encodes a variable, then $a_1 \neq a_3\wedge a_1 = a_2$.
	\end{itemize}
	If $\bar{a}$ encodes a clause, it is partitioned as follows
	\[
	\underset{a_2 = a_3}{\underbrace{a_1 a_2 a_3}}\underset{m+1}{\underbrace{a_{4}\cdots a_{m+4}}}\underset{|\bar{x}|}{\underbrace{a_{m+5}\cdots a_{m+4+|\bar{x}|}}}\underset{\text{padding elements}}{\underbrace{a_{m+5+|\bar{x}|}\cdots a_{d}}}
	\]
	where $a_1\neq a_3, a_2 = a_3$, and $a_{m+5},\dots, a_{m+4+|\bar{x}|}$ are interpretations for $\bar{x}$. $a_4,\dots, a_{m+4}$ are used to encode the clauses $\forall\bar{y}\neg X_{k+1}\bar{z}\bar{y}$, $D_1,\dots,D_m$. More precisely, we use $a_1 = a_4 \wedge \bigwedge_{5\leq j \leq m+4} a_3 = a_j$ to indicate that $\bar{a}$ encodes the clause $\forall\bar{y}\neg X_{k+1}\bar{z}\bar{y}$, and use $\bigwedge_{4\leq j \leq i+4} a_1 = a_j \wedge \bigwedge_{i+5 \leq h \leq m+4} a_3 = a_h$ to indicate that $\bar{a}$ encodes $D_i$ for $i\in \{1,\dots, m\}$, respectively. 
	This can be expressed by the formula
	\begin{equation}\label{eqt-clau1}
		\biggl( v_1 = v_4 \wedge \bigwedge_{5\leq j \leq m+4} v_3 = v_j \biggr) \vee
		\bigvee_{1\leq i \leq m} \biggl( \bigwedge_{4\leq j \leq i+4} v_1 = v_j \wedge \bigwedge_{i+5 \leq h \leq m+4} v_3 = v_h \biggr).
	\end{equation}
	We also require that there is no formula that is false in clause $D_i$ $(1\leq i \leq m)$, when $\bar{x}$ are interpreted by $a_{m+5}\cdots a_{m+4+|\bar{x}|}$.
	Let $\alpha_i$ denote the first-order part of $D_i$. This is can be expressed by the formula
	\begin{equation}\label{eqt-clau2}
		\bigvee_{1\leq i \leq m}\biggl(
		\Bigl( \bigwedge_{4\leq j \leq i+4} v_1 = v_j \wedge \bigwedge_{i+5 \leq h \leq m+4} v_3 = v_h \Bigr)
		\rightarrow \alpha_i[\bar{x}/v_{m+5}\cdots v_{m+4+|\bar{x}|}] \biggr).
	\end{equation}
	All padding elements must equal $a_1$, which can be expressed by 
	\begin{equation}\label{eqt-clau3}
		v_1\neq v_3\wedge v_2 = v_3 \wedge \bigwedge_{m+5+|\bar{x}|\leq i \leq d} v_1 = v_i
	\end{equation}
	Define $\pi_{\mathrm{Clause}}(\bar{v})$ to be the conjunction of \eqref{eqt-clau1}, \eqref{eqt-clau2} and \eqref{eqt-clau3}.
	
	If $\bar{a}$ encodes a variable, it is partitioned as follows
	\[
	\underset{a_{1}=a_{2}}{\underbrace{a_1 a_2 a_3}}\underset{k+1}{\underbrace{a_4\cdots a_{k+4}}}\underset{artiy(X_i)}{\underbrace{a_{k+5}\cdots a_{k+4+arity(X_i)}}}\underset{\text{padding elements}}{\underbrace{a_{k+5+arity(X_i)}\cdots a_{d}}}
	\]
	where $a_1\neq a_3, a_1 = a_2$, and $a_4\cdots a_{k+4}$ encode $X_1,\dots,X_{k+1}$. More precisely, we use $\bigwedge_{4\leq j \leq i+3} a_1 = a_j \wedge \bigwedge_{i+4 \leq h \leq k+4} a_3 = a_h$ to indicate that $\bar{a}$ encodes the variable with relation symbol $X_i$, where $(1 \leq i \leq k+1)$. This can be expressed by the following formula
	\[ 
	\mathrm{Var}_i(\bar{v}) = \biggl( \bigwedge_{4\leq j \leq i+3} v_1 = v_j \wedge \bigwedge_{i+4 \leq h \leq k+4} v_3 = v_h \biggr).
	\]
	We use $a_{k+5}\cdots a_{k+4+arity(X_i)}$ to indicate that $\bar{a}$ encodes the atom $X_i a_{k+5}\cdots a_{k+4+arity(X_i)}$. We also require that all padding elements $a_{k+5+arity(X_i)}\cdots a_{d}$ equal $a_1$. The formula $\pi_{\mathrm{Var}_i}(\bar{v})$, for $i \in \{1,\dots, k-1\}$, is defined by
	\[ \pi_{\mathrm{Var}_i}(\bar{v})=  \biggl( v_1\neq v_3\wedge v_1 = v_2 \wedge \mathrm{Var}_i(\bar{v}) \wedge \bigwedge_{k+5+arity(X_i)\leq j \leq d} v_1 = v_j \biggr). \]
	Define $\pi_{V_k}(\bar{v})$ to be the conjunction of $v_1\neq v_3\wedge v_1 = v_2$ and
	\[
	\biggl( \mathrm{Var}_k(\bar{v}) \wedge \bigwedge_{k+5+arity(X_k)\leq j \leq d} v_1 = v_j \biggr) \vee 
	\biggl( \mathrm{Var}_{k+1}(\bar{v}) \wedge \bigwedge_{k+5+arity(X_{k+1})\leq j \leq d} v_1 = v_j \biggr).
	\]
	
	In the following we define the formula $\pi_{\mathrm{Pos}}(\bar{v}_{1},\bar{v}_{2})$, which expresses that the atom encoded by $\bar{v}_2$ occurs positively in clause $D_j$ $(1\leq j \leq m)$ encoded by $\bar{v}_1$.
	Let $\bar{v}_1 = v_{1,1}\dots v_{1,d}$ and $\bar{v}_2 = v_{2,1}\dots v_{2,d}$.
	We use the following formula $\varphi_{D_j} (\bar{v}_1)$ to express that $\bar{v}_1$ encodes clause $D_j$ for $j\in \{1,\dots, m\}$.
	\[
	\varphi_{D_j} (\bar{v}_1) =\biggl( \pi_{\mathrm{Clause}}(\bar{v}_1)\wedge 
	\Bigl( \bigwedge_{4\leq l \leq j+4} v_{1,1} = v_{1,l} \wedge \bigwedge_{j+5 \leq h \leq m+4} v_{1,3} = v_{1,h} \Bigr) \biggr).
	\]
	Suppose that the atomic formula $X_i \bar{x}'$ occurs in clause $D_j$, where $\bar{x}'=x_1'\dots x_{arity(X_i)}'$ are variables from $\bar{x}$, and $\bar{v}_2$ encodes the atom $X_i \bar{v}_2'$ where $\bar{v}_2'$ are the corresponding elements in $\bar{v}_2$ by its definition.
	Let $\bar{v}_1'$ be obtained by replacing $\bar{x}'$ with the corresponding elements in $\bar{v}_1$ that encodes $D_j$ (note that $v_{1,m+5},\dots, v_{1,m+4+|\bar{x}|}$ are interpretations for $\bar{x}$, and $X_i \bar{x}'[\bar{x}/(v_{1,m+5},\dots, v_{1,m+4+|\bar{x}|})]= X_i\bar{v}_1'$). We require that $X_i\bar{v}_1' = X_i\bar{v}_2'$, i.e., $\bar{v}_1' = \bar{v}_2'$.
	The following formula $\alpha_{D_j, X_i}(\bar{v}_1,\bar{v}_2)$ expresses that the atom encoded by $\bar{v}_2$ occurs positively in clause $D_j$ encoded by $\bar{v}_1$.
	For $i \in \{1,\dots, k-1\}$, define
	\[
	\alpha_{D_j, X_i}(\bar{v}_1,\bar{v}_2) = \biggl( \varphi_{D_j} (\bar{v}_1) \wedge \pi_{\mathrm{Var}_i}(\bar{v}_2) \wedge \bigvee_{X_i \bar{x}' \text{ occurs positively in } D_j} \bar{v}_1'=\bar{v}_2' \biggr),
	\]
	and for $i \in \{k, k+1\}$, define
	\[
	\alpha_{D_j, X_i}(\bar{v}_1,\bar{v}_2) = \biggl( \varphi_{D_j} (\bar{v}_1) \wedge \pi_{\mathrm{Var}_k}(\bar{v}_2) \wedge \mathrm{Var}_i(\bar{v}_2) \wedge \bigvee_{X_i \bar{x}' \text{ occurs positively in } D_j} \bar{v}_1'=\bar{v}_2' \biggr).
	\]
	Define $\pi_{\mathrm{Pos}}(\bar{v}_1,\bar{v}_2)$ to be the conjunction of $v_{1,1} = v_{2,1} \wedge v_{1,3} = v_{2,3}$ and
	\[ 
	\bigvee \{\alpha_{D_j, X_i}(\bar{v}_1,\bar{v}_2)\ |\ X_i \text{ has a positive occurrence in } D_j\}.
	\]

	Similarly, we can define the formula $\pi_{\mathrm{Neg}}'(\bar{v}_1,\bar{v}_2)$ to express that the atom encoded by $\bar{v}_2$ occurs negatively in clause $D_j$ $(1\leq j \leq m)$ encoded by $\bar{v}_1$.
	For the clause $\forall\bar{y}\neg X_{k+1}\bar{z}\bar{y}$, let $\bar{v}_1''$ and $\bar{v}_2''$ be obtained by replacing $\bar{z}$ with the corresponding elements in $\bar{v}_1$ that encodes the clause, and the corresponding elements in $\bar{v}_2$ that encodes $X_{k+1}$, respectively. Let
	\[ 
	\beta_{X_{k+1}}(\bar{v}_1,\bar{v}_2) =\left(
	\begin{array}{l}
		\pi_{\mathrm{Clause}}(\bar{v}_1)\wedge (v_{1,1} = v_{1,4} \wedge \bigwedge_{5\leq i \leq m+4} v_{1,3} = v_{1,i}) \wedge \\
		\pi_{\mathrm{Var}_k}(\bar{v}_2) \wedge \mathrm{Var}_{k+1}(\bar{v}_2) \wedge v_{1,1} = v_{2,1} \wedge v_{1,3} = v_{2,3} \wedge \bar{v}_1''= \bar{v}_2''
	\end{array} \right).
	\]
	Define $\pi_{\mathrm{Neg}}(\bar{v}_1,\bar{v}_2) = \pi_{\mathrm{Neg}}'(\bar{v}_1,\bar{v}_2) \vee \beta_{X_{k+1}}(\bar{v}_1,\bar{v}_2)$.
	 
	Let $\Theta$ be the following formula
	\[
	\exists Z_1\forall Z_2\dots\exists Z_{k-1}\forall Z_k\exists Y \forall \bar{v}_1\forall \bar{v}_2 \left(\begin{array}{l}
		\exists \bar{z} Y(\bar{z})\wedge \Bigl(Y(\bar{v}_1)\rightarrow \pi_{\mathrm{Clause}}(\bar{v}_1)\Bigr)\wedge\\
		\bigwedge_{1\leq i\leq k}\Bigl(Y(\bar{v}_1)\wedge \pi_{\mathrm{Pos}}(\bar{v}_1,\bar{v}_2)\wedge \pi_{\mathrm{Var}_i}(\bar{v}_2)\rightarrow Z_{i}(\bar{v}_2)\Bigr)\wedge \\
		\bigwedge_{1\leq i\leq k}\Bigl(Y(\bar{v}_1)\wedge \pi_{\mathrm{Neg}}(\bar{v}_1,\bar{v}_2)\wedge \pi_{\mathrm{Var}_i}(\bar{v}_2)\rightarrow \neg Z_{i}(\bar{v}_2)\Bigr)
	\end{array}\right).
	\]
	The formula $\Theta$ says that there is a valuation for $X_1$, for any valuation to $X_2$, \dots, there is a valuation to $X_{k-1}$, for any valuation to $X_k$ and $X_{k+1}$, there is a nonempty set $Y$ of clauses, such that all literals in the clauses are true under the valuation.
	$\Phi$ and $\Theta$ are equivalent on the structures with at least two elements. For any finite structure, there is a quantifier-free formula that captures its isomorphism type~\cite{ebbinghaus1995}. So on one-element structures, $\Phi$ is equivalent to $\forall x \forall y(x=y\wedge \delta(x))$, where $\delta(x)$ is a disjunction of isomorphism types of one-element structures satisfying $\Phi$.
	The formulas $\Theta \vee \forall x \forall y(x=y\wedge \delta(x))$ and $\Phi$ are equivalent on all finite structures.
	Since all formulas in $\Pi$ are quantifier-free,  $\Theta \vee \forall x \forall y(x=y\wedge \delta(x))$ can be converted to an equivalent $\Sigma_{k+1}^{1}$-KROM$^{r}$ formula by elementary techniques.
\end{proof}

The following proposition says that the data complexity of SO-KROM$^r$ is in the polynomial hierarchy.

\begin{prop}\label{prop-datacomp}
	For each $k\geq 1$, if $k$ is odd, then the data complexity of $\Pi^1_{k+1}\text{-}\mathrm{KROM}^r$ and $\Pi^1_{k+2}\text{-}\mathrm{KROM}^r$ are in $\Pi^p_k$; if $k$ is even, then the data complexity of $\Sigma^1_{k+1}\text{-}\mathrm{KROM}^r$ and $\Sigma^1_{k+2}\text{-}\mathrm{KROM}^r$ are in $\Sigma^p_k$.
\end{prop}
\begin{proof}
	From Corollary~\ref{coro-sokequiv}, we know that $\Pi^1_{k+1}$-KROM$^r \equiv \Pi^1_{k+2}$-KROM$^r$ if $k$ is odd, and $\Sigma^1_{k+1}$-KROM$^r \equiv \Sigma^1_{k+2}$-KROM$^r$ if $k$ is even.
	We only prove that the data complexity of $\Pi^1_{k+1}$-KROM$^r$ is in $\Pi^p_k$ ($k$ is odd), the proof for the other cases is similar.
	
	Let $\Phi= \forall \olsi{X}_1 \exists \olsi{X}_2 \dots \forall \olsi{X}_k \exists \olsi{X}_{k+1} \forall \bar{x} \varphi$ be a $\Pi^1_{k+1}$-KROM$^r$ formula ($k$ is odd).
	Given a structure $\mathcal{A}$, we construct an alternating Turing machine that first assigns the values of $\olsi{X}_1,\olsi{X}_2, \dots, \olsi{X}_k$ alternately between universal and existential moves according to their quantifier types. 
	This step can be done in $\Pi^p_k$. The complexity of deciding whether $(\mathcal{A},\olsi{X}^A_1,\olsi{X}^A_2, \dots, \olsi{X}^A_k)\models \exists\olsi{X}_{k+1} \forall \bar{x} \varphi$ is in NL, since all occurrences of $\exists \bar{z} X_i \bar{z}$ $(1\leq i \leq k)$ in $\exists\olsi{X}_{k+1} \forall \bar{x} \varphi$ can be replaced by their truth values, and the resulting formula is a $\Sigma^1_1$-KROM$^r$ formula which can be evaluated in NL by Proposition~\ref{prop-sg1kromNL}.
	Therefore, the total complexity of checking $\mathcal{A}\models \Phi$ is in $\Pi^p_k$.
\end{proof}

Since $\Sigma^1_k$ captures $\Sigma^p_k$ and $\Pi^1_k$ captures $\Pi^p_k$ for $k\geq 1$, combining Proposition~\ref{prop-exprness} with Proposition~\ref{prop-datacomp} we conclude the following corollary.

\begin{cor}
	On all finite structures, for each $k\geq 1$, if $k$ is even, then $\Sigma^1_{k+1}\text{-}\mathrm{KROM}^r\equiv \Sigma^1_k$, and if $k$ is odd, then $\Pi^1_{k+1}\text{-}\mathrm{KROM}^r \equiv \Pi^1_k$.
\end{cor}

\begin{thm}\label{them-maincapture}
	On all finite structures, for each $k\geq 1$, if $k$ is even, then $\Sigma^1_{k+1}\text{-}\mathrm{KROM}^r$ captures $\Sigma^p_k$, and if $k$ is odd, then $\Pi^1_{k+1}\text{-}\mathrm{KROM}^r$ captures $\Pi^p_k$.
\end{thm}

\section{An extended version of second-order Krom logic}\label{sec-exspkr}
In this section, we define second-order extended Krom logic and study its expressive power and data complexity.
\begin{defi}
	Second-order extended Krom logic over a vocabulary $\tau$, denoted by SO-EKROM($\tau$), is the set of second-order formulas
	of the form
	\[
	\forall X_{1}\exists Y_{1}\cdots\forall X_{k}\exists Y_{k}\forall \bar{x}(C_{1}\wedge\cdots\wedge C_{n}),
	\]
	where $C_{i}$ $(1\leq i \leq n)$ are extended Krom clauses with respect to $Y_1,\dots,Y_k$, more precisely,
	each $C_{i}$ is a disjunction of the form
	\[
	\alpha_{1}\vee\cdots\vee\alpha_{l}\vee H_{1}\vee H_{2},
	\]
	where
	\begin{enumerate}
		\item each $\alpha_{s}$ is either $Q\bar{y}$ or $\neg Q\bar{y}$, where $Q \in\tau\cup\{X_{1},\dots,X_k \}$,
		\item each $H_{t}$ is either $Y_{i}\overline{z}$ or its negation $\neg Y_{i}\overline{z}$, where $(1\leq i \leq k)$.
	\end{enumerate}
\end{defi}

\begin{prop}
	$\mathrm{SO\text{-}EKROM}$ is closed under substructures.
\end{prop}
\begin{proof}
	All universal first-order formulas are closed under substructures. It is easy to check that the formula obtained by quantifying a relation in a formula which is closed under substructures still preserves this property.
\end{proof}

\begin{prop}\label{prop-datacompekrom}
	The data complexity of $\mathrm{SO\text{-}EKROM}$ is in $\mathrm{co\text{-}NP}$.
\end{prop}
\begin{proof}
	Let $\Phi=\forall X_1 \exists Y_1 \dots \forall X_k \exists Y_k \forall \bar{x} (C_{1}\wedge\cdots\wedge C_{n})$ be a SO-EKROM formula over $\sigma$. For an arbitrary $\sigma$-structure $\mathcal{A}$, we  replace the first-order part $\forall \bar{x}(C_{1}\wedge\cdots\wedge C_{n}) $ by 
	$\bigwedge_{\bar{a}\in A^{|\bar{x}|}}(C_1\wedge \cdots\wedge C_n)[\bar{x}/\bar{a}]$. We remove the clauses with a formula $(\neg)R\bar{b}$ that is true in $\mathcal{A}$ and delete the formulas $(\neg)R\bar{b}$ that are false in $\mathcal{A}$ in every clause, where $R$ is a relation symbol in $\sigma$. 
	Then we replace each second-order quantifier $\forall X_i$ ($\exists Y_i$) in the prefix with a sequence $\forall X_i \bar{d}_1\dots \forall X_i \bar{d}_{|A|^{arity(X_i)}}$ ($\exists Y_i \bar{d}_1\dots \exists Y_i \bar{d}_{|A|^{arity(X_i)}}$) where each $\bar{d}_j\in A^{arity(X_i)}$. We treat the atoms $X_i \bar{d}_j$ ($Y_i \bar{d}_j$) as propositional variables, the resulting formula $\psi$ is a QE-2CNF formula. It is clear that $\mathcal{A}\models \Phi$ iff $\psi$ is true. It was proved that for any fixed number $m$, the evaluation problem for the QE-2CNF formulas whose quantifier prefixes have $m$ alternations is in $\mathrm{co\text{-}NP}$~\cite{floegel1990}. Hence, whether $\mathcal{A}\models \Phi$ holds is decidable in $\mathrm{co\text{-}NP}$.
\end{proof}

Since $\Pi^{1}_{1}$ captures co-NP, we can get the following corollary.
\begin{cor}\label{cor-soekrom_pi}
	Every $\mathrm{SO\text{-}EKROM}$ formula is equivalent to a $\Pi^{1}_{1}$ formula on ordered finite structures.
\end{cor}

\begin{prop}\label{prop-pikromr_pi12ekrom}
	$\Pi^{1}_{2}\text{-}\mathrm{KROM}^r \leq \Pi^{1}_{2}\text{-}\mathrm{EKROM}$ on ordered finite structures. 
\end{prop}
\begin{proof}
	Let $\forall \olsi{X} \exists \ols{Y} \forall \bar{x}\varphi$ be a $\Pi^{1}_{2}$-KROM$^r$ formula. 
	We see that $\exists \ols{Y} \forall \bar{x}\varphi$ is a $\Sigma^1_1$-KROM$^r$ formula. By Corollary~\ref{coro-sigcapnl}, it is equivalent to a $\Sigma^{1}_{1}\text{-KROM}$ formula on ordered finite structures. This implies that $\forall \olsi{X} \exists \ols{Y} \forall \bar{x}\varphi$ is equivalent to a $\Pi^{1}_{2}\text{-EKROM}$ formula on ordered finite structures.
\end{proof}

Combining Proposition~\ref{prop-exprness}, Corollary~\ref{cor-soekrom_pi} and Proposition~\ref{prop-pikromr_pi12ekrom} gives the following corollary.
\begin{cor}
	$\mathrm{SO\text{-}EKROM} \equiv \Pi^{1}_{2}\text{-}\mathrm{EKROM} \equiv \Pi^{1}_{1}$ on ordered finite structures.
\end{cor}

\begin{thm}
	Both $\mathrm{SO\text{-}EKROM}$ and $\Pi^{1}_{2}\text{-}\mathrm{EKROM}$ can capture $\mathrm{co\text{-}NP}$ on ordered finite structures.
\end{thm}

\section{Conclusion}\label{sec-concln}
In this paper, we introduce second-order revised Krom logic and study its expressive power and data complexity. SO-KROM$^{r}$ is an extension of SO-KROM by allowing $\exists \bar{z}R\bar{z}$ in the formula matrix, where $R$ is a second-order variable.  
For SO-KROM$^{r}$, we show that the innermost universal second-order quantifiers can be removed. Hence, $\Sigma^1_k$-KROM$^{r}\equiv \Sigma^1_{k+1}$-KROM$^{r}$ for odd $k$, and $\Pi^1_k$-KROM$^{r}\equiv \Pi^1_{k+1}$-KROM$^{r}$ for even $k$. 
SO-KROM collapses to its existential fragment. The same statement is unlikely to be true for SO-KROM$^{r}$.
On ordered finite structures, we prove that $\Sigma^1_1$-KROM$^r$ equals $\Sigma^1_1$-KROM, and captures NL. On all finite structures, we show that $\Sigma^1_{k}\equiv \Sigma^1_{k+1}$-KROM$^r$ for even $k$, and $\Pi^1_{k}\equiv \Pi^1_{k+1}$-KROM$^r$ for odd $k$. This result gives an alternative logic for capturing the polynomial hierarchy, which is the main contribution of the paper. We also study an extended version of second-order Krom logic SO-EKROM.
On ordered finite structures, SO-EKROM collapses to $\Pi^{1}_{2}$-EKROM and equals $\Pi^1_1$. Therefore, both of them can capture co-NP on ordered finite structures.

\bibliographystyle{alphaurl}
\bibliography{./reference}
\end{document}